\documentclass[conference,a4paper]{IEEEtran}

\usepackage[utf8]{inputenc} 
\usepackage[T1]{fontenc}
\usepackage{url}
\usepackage{ifthen}
\usepackage{cite}
\usepackage[cmex10]{amsmath}

\usepackage{filecontents}

\usepackage[english]{babel}

\usepackage{csquotes}

\usepackage[usenames, dvipsnames]{xcolor}

\usepackage{xstring}

\usepackage{amsthm}

\usepackage{blindtext}

\usepackage{amssymb}

\usepackage{amsmath}

\usepackage{xparse}

\usepackage{stmaryrd}

\usepackage{enumitem}

\usepackage{mathtools}

\usepackage{bm}

\usepackage[parfill]{parskip}

\usepackage{xpatch}

\newenvironment{alphenum}
    {\begin{enumerate}[label=\emph{\alph*}]
    }
    { 
    \end{enumerate}
    }
\newtheorem{theorem}{Theorem}[section]

\newtheorem{proposition}[theorem]{Proposition}

\newtheorem{conjecture}{Conjecture}[section]

\newtheorem{lemma}[theorem]{Lemma}

\theoremstyle{remark}

\theoremstyle{definition}
\newtheorem{definition}{Definition}[section]

\newcommand{\ceil}[1]{\lceil#1\rceil}

\newcommand{\tand}{\textrm{ and }}
\newcommand{\tspace}{\textrm{ }}
\newcommand{\tcomma}{\textrm{, }}
\newcommand{\set}[1]{\{#1\}} 
\newcommand{\R}{\mathbf{R}}

\newcommand{\Z}{\mathbb{Z}}
\DeclareMathOperator*{\E}{\mathbb{E}}
\newcommand{\cc}{\mu^{cc}_q}
\newcommand{\wt}{\mathsf{wt}}

\newcommand{\cX}{\mathcal{X}}

\newcommand{\bcap}{\mathsf{cap}}
\newcommand{\rR}{\mathbf{R}}
\newcommand{\cR}{\mathcal{R}}
\newcommand{\rA}{\mathbf{A}}
\newcommand{\cP}{\mathcal{P}}
\newcommand{\rX}{\mathbf{X}}
\newcommand{\rY}{\mathbf{Y}}

\begin{document}
\title{Bounds on the Zero-Error List-Decoding Capacity of the $q/(q-1)$ Channel}

\author{
  \IEEEauthorblockN{Siddharth Bhandari and Jaikumar Radhakrishnan }
  \IEEEauthorblockA{Tata Institute of Fundamental Research\\
                    Homi Bhabha Road\\ 
                    Mumbai 400005, INDIA\\
                    Email: \{siddharth.bhandari, jaikumar\}@tifr.res.in}

}

\maketitle

\begin{abstract}

We consider the problem of determining the zero-error list-decoding
capacity of the $q/(q-1)$ channel studied by Elias (1988). The
$q/(q-1)$ channel has input and output alphabet consisting of $q$ symbols, say,
$\cX = \{x_1,x_2,\ldots, x_q\}$; when the channel receives an input
$x \in \cX$, it outputs a symbol other than $x$ itself. Let
$n(m,q,\ell)$ be the smallest $n$ for which there is a code
$C \subseteq \cX^n$ of $m$ elements such that for every list
$w_1, w_2, \ldots, w_{\ell+1}$ of distinct code-words from $C$, there
is a coordinate $j \in [n]$ that satisfies $\{w_1[j], w_2[j], \ldots,
w_{\ell+1}[j]\} = \cX$. We show that for $\epsilon<1/6$, for all large $q$ and large enough $m$, $n(m,q, \epsilon q\ln{q}) \geq \Omega(\exp{(q^{1-6\epsilon}/8)}\log_2{m})$.

The lower bound obtained by Fredman and Koml\'{o}s (1984) for perfect hashing implies
that $n(m,q,q-1) = \exp(\Omega(q)) \log_2 m$; similarly,
the lower bound obtained by K\"{o}rner (1986) for nearly-perfect hashing
implies that $n(m,q,q) = \exp(\Omega(q)) \log_2 m$. These results show that the
zero-error list-decoding capacity of the $q/(q-1)$ channel with lists
of size at most $q$ is exponentially small. Extending these bounds,
Chakraborty \emph{et al.} (2006) showed that the capacity remains
exponentially small even if the list size is allowed to be as large as
$1.58q$. Our result implies that the zero-error list-decoding capacity
of the $q/(q-1)$ channel with list size $\epsilon q$ for $\epsilon<1/6$ is $\exp{(\Omega(q^{1-6\epsilon}))}$. This resolves the conjecture raised by Chakraborty \emph{et al.} (2006) about the zero-error list-decoding capcity of the $q/(q-1)$ channel at larger list sizes.
\end{abstract}

\section{Introduction}
We study the zero-error-list-decoding capacity of the $q/(q-1)$
channel. The input and output alphabet of this channel are a
set of $q$ symbols, namely $\cX = \{x_1, x_2, \ldots, x_q\}$;
when the symbol $x \in \cX$ is input, the output symbol can be
anything other than $x$ itself. We wish
to design good error correcting codes for such a channel.  For the
$q/(q-1)$ channel it is impossible to recover the message without
error if the code has at least two code-words: in fact, no matter how
many letters are used for encoding, for every set of up to $(q-1)$ input
code-words, one can construct an output word that is compatible
with \emph{all} of them. It is, however, possible to design codes
where on receiving an output word from the channel, one can narrow down
the input message to a set of size at most $(q-1)$---that is, we can
\emph{list-decode} with lists of size $(q-1)$. 
Such codes have rate exponentially small in $q$.
~\\

\begin{definition}[Code, Rate]
\label{Defn_Rate}
A code $C \subseteq \{x_1,\ldots,x_q\}^n$ is 
an $\ell$-list-decoding-code for the $q/(q-1)$ channel, if for every
output word
$\sigma' \in \cX^n$, we have
$\bigl|\{ \sigma \in \cX^n: \mbox{the input word $\sigma$ is compatible with $\sigma'$} \}\bigr| \leq \ell$.
Let $n(m,q,\ell)$ be the smallest $n$ such that there exists an
$\ell$-list-decoding code for the $q/(q-1)$ channel with $m$
code-words. The zero-error-list-of-$\ell$-rate of $C$, {$|C|=m$}, is given by
$\frac{1}{n} \log_2 (m/\ell)$, and the list-of-$\ell$-capacity of the
$q/(q-1)$ channel, denoted by $\bcap (q,\ell)$, is the least upper bound
on the attainable zero-error-list-of-$\ell$-rate across all  $\ell$-list-decoding-codes.
\end{definition}

The list-of-$2$-capacity of the $3/2$ channel was studied by
Elias~\cite{E88}, who showed that $0.08 \approx\log_2 (3) -
1.5 \leq \bcap(3,2) \leq \log_2(3) -1 \approx 0.58$.  For the $4/3$
channel, Dalai, Guruswami and Radhakrishnan~\cite{DGR2016} showed that
$\bcap(4,3) \leq 6/19 \approx 0.3158$, improving slightly on an earlier upper bound of 0.3512 shown by Arikan~\cite{Arikan94}; it was shown by K\"{o}rner
and Marton~\cite{KM88} that  $\bcap(4,3)
\geq (1/3)\log_2 (32/29)\approx 0.0473$. For general $q$, one can obtain the
following upper bound using a routine probabilistic argument.
\begin{proposition}
$n(m,q,q-1) = \exp(O(q))\lg m$.
\end{proposition}
This implies that the $\bcap(q,q-1) = \exp(-O(q))$. So for each fixed
$q$ we do have codes with positive rate, but the rate promised by this
construction goes to zero exponentially with $q$. Fredman and
Koml\'{o}s~\cite{FK84} showed that this exponential deterioration is
inevitable; K\"{o}rner showed that $\bcap(q,q)
= \exp(-\Omega(q))$. On the other hand, it can be shown that
$\bcap(q,\ceil{q \ln q}) = 1/q$, and that for all functions
$\ell: \Z \rightarrow \Z$ we have $\bcap(q,\ell(q)) \geq 1/q$. Thus,
the list-of-$\ell$-capacity of the $q/(q-1)$ channel cannot be better
than $1/q$ unless $\ell$ is allowed to grow with $m$.

We thus have the following situation. The list-of-$\ell$-rate of any
code reaches the optimal value of $1/q$ when the list-size is about
$q \ln q$; however, the list-of-$(q-1)$ (as well as list-of-$q$) rate is
exponentially small in $q$. It is interesting, therefore, to study the
trade-off between the list size and the rate, and determine how
the rate changes from inverse polynomial in $q$ to exponentially small
in $q$. Chakraborty, Radhakrishnan, Raghunathan and
Sasatte~\cite{CRRS} addressed this question and showed the following.

\begin{theorem} \label{thm:Chetal}
For every $\epsilon > 0$, there is a $\delta > 0$ such that for all
large $q$ and large enough $m$, we have $n(m,q,(\eta
- \epsilon)q) \geq \exp(\delta q) \log_2 m$, where $\eta =
e/(e-1) \approx 1.58$. Thus, $\bcap(q,(\eta-\epsilon)q) = \exp(-\Omega(q))$.
\end{theorem}

We show the following.

\begin{theorem}[\textbf{Result}]
\label{Result1}
For every $\epsilon<1/6$, for all large $q$ and large enough $m$, we have 
$n(m,q, \epsilon q\ln{q}) \geq \Omega(\exp{(q^{1-6\epsilon}/8)}\log_2{m})$.
Thus, for all $\epsilon<1/6$, $\bcap(q,\epsilon q\ln{q}) =\exp(-\Omega(q^{1-6\epsilon}))$.
\end{theorem}

This establishes both parts of the conjecture of Chakraborty \emph{et al.}
which states the following.
\begin{conjecture}
\label{Conjecture_CRRS}
 
\emph{(a)} For all constants $c>0$, there is a constant $\alpha$, such that for all large $m$, we have $n(m,q,cq)\geq \exp{(\alpha q)}\log_2 m$.\\
\emph{(b)} For all functions $\ell(q) = o(q \log_2 q)$ and all large
$m$, we have  $n(m,q,\ell(q)) \geq q^{\omega(1)} \log_2 m$.

\end{conjecture}

\subsection{Overview of our approach}

We extend the approach of Chakraborty \emph{et al.}, which in turn
was based on the approach used by Fredman and Koml\'{o}s~\cite{FK84} to
obtain lower bounds on the size of families of perfect hash
functions. To describe our adaptation of this approach, it will be
convenient to reformulate the problem using matrix terminology.

Consider $C \subseteq \cX^n$ with $m$ code-words. We can build an
$m \times n$ matrix $C = (c_{ij}: i=1,\ldots, m \mbox{ and } j=1,\ldots,
n)$ (we use the name $C$ both for the code and the associated matrix)
by writing the code-words as rows of the matrix (the order does not matter): so $c_{ij}
= k$ iff the $j$-th component of the $i$-th code-word is $x_k\in \cX$.  Then,
$C$ is an $\ell$-list-decoding code iff the matrix has the following
property: for every choice $R$ of $\ell + 1$ rows, there
is a column $h$ such that $\{c_{rh}: r \in R\} = \cX$. In this
reformulation, $n(m,q,\ell)$ is the minimum $n$ so that there exists a matrix with this property. We refer to such a matrix as an
$\ell$-list-decoding matrix. Furthermore, instead of writing $c_{rh}$
we write $h(r)$; indeed, in the setting of hash families (originally considered by Fredman and Koml\'{o}s), the columns correspond to hash functions that assign a symbol in $\cX$ to each row-index in $[m]$.

We can now describe the approach of Chakraborty \emph{et al.} Fix a list-size  
$\ell= \alpha q$. Suppose there is an $\ell$-list decoding matrix $C$ with
$n = \exp(\beta q) \log_2 m$ columns. We wish to show that if $\beta$ is
small then the matrix cannot have the required property; that is,
we can find a set $R$ of $\ell+1$ rows for which $h(R)$ is a proper
subset of $[q]$ for every column $h$. To exhibit such a set $R$ we
will proceed in stages. In the first stage, we pick a subset $R_1$ of
$q-2$ rows at random. Consider a column $h$. What can we expect?  We
expect to see a good number of \emph{collisions}, where 
the same symbol appears in column $h$ at two different rows in $R_1$. In fact, we
expect $h(R)$ to contain only about $q(1-1/e)$ elements. By appealing to
standard results (e.g., McDiarmid's inequality), we may conclude that
with probability exponentially close to $1$ (that is, of the form
$1-\exp(-\gamma q)$), $h(R)$ is unlikely to have significantly more
elements. So we might settle on a choice of $R$, so that $h(R)$
deviates significantly (say by $\epsilon q$ for some small $\epsilon$)
for at most $\exp(-\gamma q) \exp(\beta q) \log_2 m$ columns. If the original
$\beta$ is chosen to be much smaller than $\gamma$, this number 
is an exponentially small fraction of $\log_2 m$. 

The key idea now is to make these exceptional columns ineffective. We
do this by focusing our attention on a reduced number of rows. For
each exceptional column, we pick the symbol that appears most often in
that column, and restrict attention to those rows that have this symbol in the exceptional column. This depletes the number of rows by a factor at $1/q$ for
each exceptional column; after we do this sequentially for all the
$\exp(-(\gamma-\beta)q) \log_2 m \ll \log_2 m$ rows, we will be left with
$m'$ rows, where $\log_2 m' = \Omega(\log_2 m)$. We may now add more
rows to our existing list $R_1$. If we choose these from the set of
$m'$ rows, we are in no danger from the exceptional columns; in the
other columns $R_1$ spans about $q(1-1/e)$ symbols, so we can add to
$R_1$ about $q/e$ rows $R_2$ (picked from the $m'$-rows) and
still ensure that in no column $h$, we are in danger of $h(R_1\cup
R_2)$ becoming $\cX$. It is clear that we can carry this approach
further, e.g., by picking $R_2$ randomly, expecting a significant
number of internal collisions, making the exceptional columns ineffective,
focusing attention on a smaller but still significant number of rows,
etc., then picking $R_3$ from the rows that survive, and so on.  In
fact, Chakraborty \emph{et al.} derived Theorem~\ref{thm:Chetal} using
precisely this approach.

In this paper, we follow the approach outlined above 
but implement the idea more precisely. Before we describe our
contribution it will be useful to pin-point where the calculations in
Chakraborty \emph{et al.} were sub-optimal. We argued above that after
$R_1$ is picked, we expect to span only about $q(1-1/e)$ symbols in a
given column $h$. What about after $R_2$ is picked? {$R_1 \cup R_2$ contains a total of
$q + q/e$ rows: if} all symbols in column $h$ appeared with the same
frequency (and continued to do so in the $m'$ rows after the
exceptional columns were eliminated), then we should expect
$h(R_1 \cup R_2)$ to span about $(q + q/e) (1 - \exp(1 + 1/e))$
symbols. Notice that this is roughly the expected number of distinct coupons collected
in the classical \emph{coupon collector problem} after $q+q/e$
attempts. Unfortunately, there are technical difficulties that arise in 
claiming that this number will be reflected in our process
because \emph{(i)} $R_1$ and $R_2$ are not picked independently, and 
\emph{(ii)} even if the symbols appeared with the same frequency initially, 
they may not do so after we focus on a depleted set of rows. Faced with these
difficulties, Chakraborty \emph{et al.} settled for less. Instead of
matching the bound suggested by the coupon collector problem, when
analysing the expected size of $h(R_1 \cup R_2)$, they estimated
$h(R_2)$ separately and bounded $|h(R_1 \cup R_2)|$ by $|h(R_1)| +
|h(R_2)|$, thereby ignoring $h(R_1 \cap R_2)$. The loss 
in precision resulting from the use of this union bound increases as the number of phases 
increases. Indeed, when the coupon
collector process is carried in phases by picking sets $R_1,
R_2, \ldots, R_t$ for a large $t$, progress in collecting coupons is
retarded more by collisions across sets (because for some $i\neq j$,
$h(R_i)$ and $h(R_j)$ have elements in common) than by collisions
within some $h(R_i)$. By neglecting collisions across phases, and by 
failing to track the \emph{coupon collector} process
closely, the argument in Chakraborty \emph{et al.} were unable to push the 
list size in Theorem~\ref{thm:Chetal} beyond $e/(e-1)$.

\subsubsection*{\textbf{What is new?}}
We attempt to track the progress of the
coupon collector faithfully. Instead of the set $R_1$ of size
$q-2$ that was picked earlier, we pick an ensemble (a collection of sets) $\cR^1$ of sets of
size $q-2$. Similarly, in the later steps we will pick ensembles
$\cR^2, \cR^3, \ldots$. However, in the end we pick one set $R_i$
from each of the ensembles {$\mathcal{R}^i$} respectively, and assemble our list of rows: $R_1 \cup
R_2 \cup \cdots \cup R_t$. That this process is more
effective in bounding $|h(R_1 \cup R_2 \cup \ldots \cup R_t)|$ will be
formally verified in later sections. For now, let us qualitatively see
how it helps in bounding $|h(R_1 \cup R_2)|$. We pick $\cR^1$ at
random: if the number of sets in the ensemble is large enough (we will
set it to be $\exp(\Theta(q))$), then it should reflect a random set of rows that was obtained by picking rows independently $(q-2)$-times from the set of all rows. Fix a choice for $R_2$, the set to be picked at the second stage. Consider $\rX= |h(R_1 \cup R_2)|$ where $R_1$ is picked uniformly from the ensemble $\cR^1$; let $\rY=|h(\rR_1 \cup R_2)|$, where $\rR_1$ is picked uniformly from the
set of all rows. Then, we expect $X$ and $Y$ to have similar
distribution. 
So, we proceed as follows. We pick an ensemble $\cR^1$
at random. If for a certain column $h$, the ensemble $\cR^1$ fails to
deliver a good sample, we will need to make that column ineffective as
before. Further, if some set in $\cR^1$ spans a significantly larger
number of symbols in some column, we will again make that column ineffective. After this, we pick $R_2$ from the remaining rows. We expect it to not only have a good number of internal collisions but also be such that
$|h(\rR_1 \cup R_2)|$ and $|h(\rR_1 \cup \rR_2)|$ (where the set ${\rR}_2$ is chosen uniformly from the available rows) are similar in
expectation. Now, since we ensured that the ensemble $\cR^1$ was good
for column $h$, a random choice of $R_1$ from the ensemble will
deliver a value of $|h(R_1 \cup R_2)|$ that, with high probability, can be
bounded by the number of distinct coupons picked up at the same stage by the coupon collector; in particular, it accounts for symbols common to $h(R_1)$
and $h(R_2)$. The outline above illustrates the advantages of picking
an ensemble instead of committing to just one randomly chosen set.  However,
a large ensemble comes with its drawbacks. We need to ensure that no
set in the ensemble spans too many elements in any column, or
rather, we need to eliminate any column where some set spans many
elements. This forces a more drastic reduction in the number of rows
than before (that is, now $m'$ when compared with $m$ is much smaller than
in the calculation in~\cite{CRRS}). Thus, it is important to keep the sizes of the ensembles small. The trade-off between these opposing concerns
needs to be handled with some care. The argument is presented in
detail below.

\section{Proof of the Result}

In what follows we assume that $q$ is a large natural number and $m\to \infty$.

We will need the following concentration result due to McDiarmid (1989).
\begin{lemma}[McDiarmid]
\label{McDiarmid}
Let $X_1, X_2, \ldots, X_n$ be independent random variables where each $X_k$ takes values in a finite set $A$. Let $f: A^n\to \mathbb{R}$ be such that $ {|f(x)-f(y)|\leq c}$ whenever $x$ and $y$ differ in only one coordinate. Let $Y=f(X_1, X_2,\ldots, X_n)$; then, for all $t>0$, 
\begin{align*}
    \Pr[\E[Y]-Y\geq t], \Pr[Y-\E[Y]\geq t]\leq \exp{\left(\frac{-2t^2}{nc^2}\right).}
\end{align*}
\end{lemma}

Let $C$ be an $\ell$-list-decoding-code for the $q/(q-1)$ channel with $\ell < q\ln{q}/6$.
As mentioned in the introduction, we will view $C$ as an $m\times n$ matrix with entries from $[q]$. In other words, the rows are indexed by code-words and the columns are indexed by hash functions. 
Let $\wt$ be a function from $[q]$ to $\set{0,1}$; for $A\subseteq [q]$, let $\wt(A)\coloneqq \sum_{a\in A}\wt(a)$. 
Let $\R$ be a random variable taking values in $\cP([m])$. Sometimes we use $\R$ to also refer to the distribution of this random variable.

Following the idea mentioned in the introduction, we intend to keep an ensemble $\mathcal{R}$ of sets of rows such that when we pick a new set of rows $R_2$ from a depleted number of rows $m'$, we not only observe the correct number of internal collisions within $R_2$ but also observe the correct number of collision between members of $\mathcal{R}$ and $R_2$.
This motivates the following definition.

\begin{definition}[Sampler]
\label{Defn_Sampler}
We say that an ensemble $\mathcal{R}=(R_1, R_2, \ldots, R_L)$, where each $R_i\subseteq [m]$, is a $(\gamma, \delta)$-sampler for $\R$ wrt column $h$ if $(A_1, A_2, \ldots, A_L)\coloneqq (h(R_1), h(R_2),\ldots, h(R_L))$ satisfies $\forall \wt:[q]\to \set{0,1}$
\begin{align*}
    &\Pr_{j\in_\mathbf{u}[L]}\biggl[\Bigl|\wt(A_j)-\E{\Bigl[\wt(h(\R))\Bigr]}\Bigr|\geq \gamma q\biggr]\leq \exp{(-\delta q)}.
\end{align*}
\end{definition}

The definition makes provision for all functions $\wt$, because it tries to anticipate the appropriate internal collisions
(see Lemma \ref{Composition_Lemma}) with very little advance knowledge of what the distribution on $[q]$ looks like in column $h$ after a large number of rows have been discarded.

Let $\pi:S \rightarrow [0,1]$ be a probability mass function on a finite set $S$. Let $k\geq 1$, and let $X_1,X_2, \ldots, X_k$ be independent random
variables each distributed according to $\pi$. Then, let $\pi^{\{k\}}$
denote the probability mass function of the set $\{X_1,X_2, \ldots,
X_k\}$. 

For distributions $\mathbf{A}$ and $\mathbf{B}$ on $\cP{([m]})$, let $\mathbf{A} \vee \mathbf{B}$ be the distribution of $S\cup T$ where $S\sim \mathbf{A}$ and $T \sim\mathbf{B}$, with $S$ and $T$ chosen independently. The following lemma will be the main workhorse for our argument.
\begin{lemma}[Ensemble Composition Lemma]
\label{Composition_Lemma}
Let $\R$ be a distribution on $\cP([m])$ and let $D$ be a distribution on $[m]$.
Let $\mathcal{R}$ be a $(\gamma, \delta)$-sampler for $\R$ wrt a column $h$; let $(R_1, R_2, \ldots, R_s)$ be obtained by taking $s$ independent samples from the ensemble $\mathcal{R}$. Similarly, let $\mathcal{R}'=(R'_1, R'_2, \ldots, R'_{s})$ be obtained by taking $s$ independent samples according to $\R'\sim D^{\set{tq}}$ where $t<1 $. Let $\gamma',\delta' > 0$ be such that $\delta \leq 2(\gamma')^2/t$ and $\delta > \delta'$ Let $s=\exp(\delta - \delta') q\tcomma \widetilde{\gamma}=\gamma+\gamma' \tcomma \widetilde{\delta}=\delta - \delta'$. Then, with probability $1-12\exp{(-\delta'q)}$ over the random choices, the composed ensemble 
\[ \widetilde{\mathcal{R}} \coloneqq \bigl(R_{1}\cup R'_1, R_{2}\cup R'_2,\ldots, R_{s}\cup R'_{s})\]
of cardinality $s$, is a $(\widetilde{\gamma}, \widetilde{\delta})$-sampler for $\R \vee \R'$ wrt the column $h$, and furthermore $\forall i \in [s]$,
\begin{equation}  
\Bigl|\bigl|h(R_{i}\cup R_i')\bigr|-\E{\Bigl[\bigl|h(\R\vee\R')\bigr|\Bigr]}\Bigr| \leq \widetilde{\gamma}q. \label{furthermore}
\end{equation}  
Note that this ensemble is generated according to the product distribution $\left(\mathcal{R} \vee \R'\right)^s$.
\end{lemma}

\begin{proof}

Fix $f:[q]\to \set{0,1}$ and let $\mu_f\coloneqq \E{\bigl[f(h(\R \cup \R'))\bigr]}$; similarly, for $R'\subseteq[m]$ let $\mu_f(R')\coloneqq \E_{\R}{\bigl[f(h(\R \cup R'))\bigr]}$.
First, we bound the probability that when $\mathcal{R}'$ is chosen according to $\mathbf{R}'$, it fails to have $\mu_f(R')$ close to $\mu_f$.
Using McDiarmid's inequality over the $tq$ primitive choices for $R'$, we have
\begin{align}
    \Pr_{R'\sim\R'}{\bigl[|\mu_f(R')-\mu_f|\geq \gamma' q\bigr]}&\leq 2\exp{\left(\frac{-2(\gamma')^2q^2}{tq}\right)} \nonumber\\
    &= 2\exp{\left(\frac{-2(\gamma')^2q}{t}\right)}. \label{McD:application}
\end{align}

Now, let $\wt:[q]\to {0,1}$ be defined by $\wt(x) = f(x)$ if $x \not \in h(R')$ and $\wt(x)=0$ otherwise. Then, for $R\subseteq[m]$, we have, $f(h(R\cup R')) = f(h(R')) + \wt(h(R))$. Therefore (note here $R'$ is fixed and $R$ varies randomly in $\mathcal{R}$),
\begin{align*}
    &\Pr_{R\in_\mathbf{u}\mathcal{R}}{\biggl[\Bigl|f(h(R\cup R'))- \mu_{f}(R')\Bigr|\geq \gamma q \biggr]}\\ 
    &\hspace{5em}= \Pr_{R\in_\mathbf{u}\mathcal{R}}{\biggl[\Bigl|\wt(h(R)) - \E{\bigl[\wt(h(\R))\bigr]} \Bigr|\geq \gamma q \biggr]}
    \intertext{and since $\mathcal{R}$ is a $(\gamma,\delta)$-sampler wrt $\R$, we have}
    &\Pr_{R\in_\mathbf{u}\mathcal{R}}{\biggl[\Bigl|\wt(h(R)) - \E{\bigl[\wt(h(\R))\bigr]} \Bigr|\geq \gamma q \biggr]}\leq\exp({-\delta q}). 
\end{align*}
Thus,
\begin{align}
   &\Pr_{R\in_\mathbf{u}\mathcal{R}, R'\sim \R'}{\Bigl[\bigl|f(h(R\cup R'))-\mu_f\bigr|\geq (\gamma+\gamma') q \Bigr]}\leq \nonumber\\
    &\Pr_{R\in_\mathbf{u}\mathcal{R}, R'\sim \R'}{\Bigl[|\mu_f-\mu_f(R')|\geq \gamma'  q \Bigr]} +\nonumber \\
    &\hspace{1em}\Pr_{R\in_\mathbf{u}\mathcal{R}, R'\sim \R'}{\Bigl[\bigl|\mu_f(R')-f(h(R\cup R'))\bigr|\geq \gamma q\Bigr]}\nonumber \\
    &\leq \Bigl[2\exp{\left(\frac{-2(\gamma')^2q}{t}\right)} + \exp{\bigl(-\delta q\bigr)}\Bigr] \leq 3\exp(-\delta q). \label{Sampler:appl}
\end{align}
(We used $\delta \leq 2(\gamma')^2/t$ to justify the last inequality.) 
Let $\Delta:=3\exp(-\delta q)$, the quantity on the right in (\ref{Sampler:appl}).
By taking $f$ to be the all-$1$'s function, we conclude from
(\ref{Sampler:appl}) that for each $i$ with probability at least $1-\Delta$, $\Bigl|\bigl|h(R_{i}\cup R_i')\bigr|-\E{\bigl[|h(\R\vee\R')|\bigr]}\Bigr| \leq (\gamma + \gamma') q$. 

Now, $f(h(R_i\cup R'_i)) = \bigl|h(R_i)\cup h(R_i')\bigr|$, and $\mu_{f} = \E{\bigl[|h(\R\vee\R')|\bigr]}$. 
Now, by a union bound over the $s$ choices for $i$, we obtain
\begin{align}
    &\Pr_{\widetilde{\mathcal{R}}}{\biggl[\exists i\in [s]\tcomma \Bigl|\bigl|h(R_i\cup R_i')\bigr|-\E{\bigl[\bigl|h(\R\vee\R')\bigr|\bigr]}\Bigr|\geq (\gamma + \gamma') q \biggr]} \nonumber \\
    &\hspace{8em}\leq \Delta s \leq 3\exp(-\delta' q) . \label{uniontailbound}
\end{align}  
This establishes (\ref{furthermore}).

It remains to establish our first claim that whp the ensemble picked according to  $(\mathcal{R} \vee \R')^s$ is a $(\widetilde{\gamma}, \widetilde{\delta})$-sampler
for $\R \vee \R'$.
Fix $f: [q]\rightarrow \{0,1\}$. Now, (\ref{Sampler:appl})
implies that for each $i\in[s]$, the probability that $|f(h(R_i\cup R_i'))-\mu_f|\geq (\gamma+\gamma') q$ is exponentially small in $q$.
Then, the tail probabilities for
$\mathbb{Y}\coloneqq \sum_{i=1}^s \mathbb{I}\bigl[|f(h(R_i \cup R_i'))-\mu_f|\geq (\gamma+\gamma') q\bigr]$ can be bounded by considering $\mathbf{Bin}(s,\Delta)$. Therefore, 
\begin{align}
    &\Pr_{\widetilde{\mathcal{R}}}{[\mathbb{Y}>\exp{(-\widetilde{\delta}q)}s]} \nonumber \\
    &\leq \binom{s}{\exp{(-\widetilde{\delta}q)}s}(\Delta)^{\exp{(-\widetilde{\delta}q)}s}\nonumber \\
    &\leq (e\exp{(\widetilde{\delta}q)}\Delta)^{\exp{(-\widetilde{\delta}q)}s}\nonumber \\
    &\leq 9\exp{(-\delta'q)}. \label{Jaggitailbound}
\end{align}

(We need to take a union bound against the $2^q$ possible functions $f:[q]\to \set{0,1}$: by changing $s$ to $qs$ we may easily establsih this.)
By (\ref{uniontailbound}) and (\ref{Jaggitailbound}), the probability that our ensemble fails to be a $(\widetilde{\gamma}, \widetilde{\delta})$-sampler, with $\widetilde{\gamma}= \gamma+\gamma'$ and $\widetilde{\delta}= \delta-\delta'$, or fails to satisfy (\ref{furthermore}) is at most
$12\exp{(-\delta'q)}$. 
\end{proof}

Let us recall the template of our argument. At any stage we will have an ensemble of sets of rows, say $\mathcal{R}$, and a universe $U\subseteq[m]$ to choose sets of rows from to add to $\mathcal{R}$. We will add a specific number of randomly chosen sets of rows of a particular size from $U$ and then declare those columns bad where the modified $\mathcal{R}$ deviates from its expected behaviour. Consider a set $R\in \mathcal{R}$: we want to say that the coupon-collector process at $|R|$ probes into $[q]$ is the gold standard for good behaviour, i.e., no set in $\mathcal{R}$ will have expansion more than the coupon-collector at the same stage. 
The expected number of elements that the coupon-collector process picks up after $a$ i.i.d. uniform probes into $[q]$ is approximately $q\left(1 - \exp(-a/q)\right)$: we will denote this as $\cc(a)$.
So, we need the following lemma, which is proved in the appendix.  
\begin{lemma}[Phased Coupon Collector]
\label{PCC_Lemma}
Let $a_1, a_2, \ldots, a_k$ be positive integers; let
$a=a_1+a_2+ \cdots + a_k$, and let $\pi_1, \pi_2, \ldots, \pi_k$ be
probability mass functions.  Let $\rA_1, \rA_2, \ldots, \rA_k$ be
independent random variables taking values in $\cP([q])$, where
$\rA_i \sim \pi_i^{\{a_i\}}$. Suppose $a\leq \epsilon q\ln{q}$ and $k\leq e q^\epsilon$ for some $\epsilon<1/3$, then,
\begin{align*}
     \E[|\rA_1 \cup \rA_2 \cup \cdots \cup \rA_k|] \leq 
q\left(1 - \exp(-a/q)\right) + o(q^{1-\epsilon}) \\
= \cc(a) + o(q^{1-\epsilon}).
\end{align*}
\end{lemma}
\newcommand{\proofofphasedcc}{
\begin{proof}[Proof of Lemma~\ref{PCC_Lemma}]
Consider a constant $\lambda \ll \epsilon$.
For $i =1,2,\ldots,k$, let $B_i$ be the set of $q^{1-2\epsilon-\lambda}$ elements of
$[q]$ taking the topmost values in $\pi_i$. Let $B = B_1 \cup
B_2 \cup \cdots \cup B_k$; note that $|B| \leq k q^{1-2\epsilon-\lambda} =
o(q^{1-\epsilon})$. Then, $\E[|\rA_1 \cup \rA_2 \cup \ldots \rA_k|]$ is at most
\[|B| + \sum_{x\not\in B} \left(1- \prod_{i=1}^k (1-\pi_i(x))^{a_i}\right).\]
Now, for $x \not\in B$, we have $\pi_i(x) \leq 1/q^{1-2\epsilon-\lambda}$, and 
\begin{align*}
1 - \pi_i(x) \geq \exp\bigl(-\pi_i(x)/\bigl(1-\pi_i(x)\bigr)\bigr)\\
\geq \exp\bigl(-\pi_i(x) \bigl(1+2/q^{1-2\epsilon-\lambda}\bigr)\bigr).
\end{align*}
Then, by the AM-GM inequality we have the upper-bound
\begin{align*}
|B| + q - q \exp\left(- (1+2/q^{1-2\epsilon-\lambda})(1/q)\sum_{i,x} a_i \pi_i(x)  \right).    
\end{align*}

Our claim follows from this because $\exp\left(- (1+2/\sqrt{q})(1/q)\sum_{i,x} a_i \pi_i(x)  \right) \geq \exp(-a/q)-o(1/q^{1-2\epsilon})\geq \exp{(-a/q) - o(q^{1-\epsilon})/q}$.
\end{proof}
}
Our next target is to understand the number of iterations we wish to perform, i.e., the number of times we need to enlarge the sizes of the sets surviving the ensemble $\mathcal{R}$ so that the list size hits the target of $\epsilon q\ln{q}$, where $\epsilon <1/6$. 
At the first stage we will pick up sets of rows of size about $\ell_1=q$, and expect the image size to be close to $\cc(\ell_1)$; we then prune out the exceptional columns. In the next stage, we pick sets of size about $\ell_2=q-\cc(\ell_1)$ and expect the combined image size to be close to $\cc(\ell_1+\ell_2)$. Hence, in the third iteration we pick sets of size close to $\ell_3=q-\cc(\ell_1+\ell_2)$, and so on for the subsequent iterations. We are interested in the list size after $k$ iterations, i.e, $\ell_{\leq k} \coloneqq \sum_{i=1}^{k}\ell_i$. We have the following proposition, which is proved in the appendix.

\begin{proposition}
\label{Target_Proposition}
Let $\ell_1=q$, and for $i\geq 1$ let $\ell_{i+1}= q- \cc(\sum_{j=1}^{i}\ell_j)$. 
Suppose $k=eq^\epsilon$ for some $\epsilon<1$, then, $\ell_{\leq k}\geq \epsilon q\ln{q}$.
\end{proposition}
\begin{proof}
(The series $\set{\ell_{\leq k}}$ tends to $q\ln{q}$.)
\end{proof}
\newcommand{\proofofTargetProposition}{
\begin{proof}[Proof of Proposition~\ref{Target_Proposition}]
Suppose $\ell_{\leq i}\in [jq,(j+1)q]$ for some $j\geq 0$, then, $\ell_{i+1}\geq q/e^{j+1}$. Therefore, the number of $i$'s for which $\ell_{\leq i}\in[jq,(j+1)q]\leq e^{j+1}$. Suppose $\ell_{\leq k}< \epsilon q\ln{q}$, then as a contradiction we have
\begin{align*}
    k< e + e^2 + \cdots + e^{\epsilon \ln{q}} \leq e q^\epsilon.
\end{align*}
\end{proof}
}

Finally, we need a lemma where we glue all the steps mentioned in the introduction. At each iteration $k$, we maintain an ensemble $\mathcal{R}^k$ satisfying the requisite properties.

We call a distribution $\mathbf{D}$ on $\cP({[m]})$ a $(g_1,\ldots,g_k)$-phased coupon collector distribution if $\mathbf{D}=D_1^{\set{g_1}}\vee D_2^{\set{g_2}}\ldots\vee D_k^{\set{g_k}}$ where each $D_i$ is a probability mass function on $[m]$. The following lemma tracks how the parameters change with each iteration.

\begin{lemma}[Iteration Lemma]
\label{Iteration_Lemma}
Let $k\leq q^\epsilon$ for some $\epsilon<1/5$. Let $\gamma = \gamma' = q^{-2\epsilon}/2$ and $\delta' = q^{-5\epsilon}/4$. Assume $n\leq \exp{(\delta'q)}\log_2{m}/(48\cdot q^\epsilon \log_2{q})$. Then, there exists a partition $\mathcal{H}_1(k)\sqcup\mathcal{H}_2(k)$ of the columns of $C$, a universe of rows $U_k\subseteq [m]$, an ensemble $\mathcal{R}^k=(R_1, R_2, \ldots R_{L_k})$, integers $(g_1,\ldots,g_k)$ and a $(g_1,\ldots,g_k)$-phased coupon collector distribution $\mathbf{D}_k$ such that: 
\begin{alphenum}
\item $g_1=q-2$, and $g_{i+1}=q-\cc(g_i)-(i+1)\gamma q -2$
\item $\forall i\in [L_k]\tcomma |R_i|= g_{\leq k} \geq \ell_{\leq k} -2k -k^2\gamma q/2$ 
\item $\forall h \in \mathcal{H}_2(k)\tcomma \forall i \in [L_k]\tcomma |h(R_i\cup U_k)|\leq q-1$
\item $\forall h\in \mathcal{H}_1(k)$, $\mathcal{R}^k$ is a $((k+1)\gamma,\gamma^2-k\delta')$-sampler for $\mathbf{D}_k$ wrt $h$   
\item $\forall h\in \mathcal{H}_1(k)\tcomma \forall i\in [L_k]\tspace \bigl||h(R_i)|-\E{[h(\mathbf{D}_k)]}\bigr|\leq (k+1)\gamma q$
\item $\log_2{|U_k|}\geq \log_2{m} - k\log_2{q}\cdot 24\exp{(-\delta'q)}n$.
\end{alphenum}
\end{lemma}

\begin{proof}
We will use induction on $k$.
For $k=1$ we have $g_1=q-2$. We use Lemma~\ref{Composition_Lemma} with $\R$ being the constant $\emptyset$, and $\mathcal{R}=\set{\emptyset}$. Clearly, $\mathcal{R}$ is a $(\gamma,\gamma^2)$-sampler for $\R$. Let $D$ be the uniform distribution over $[m]$ and let $\mathcal{R}'=(R'_1, R'_2, \ldots, R'_{s})$ be obtained by taking $s=\exp{((\gamma^2 -\delta')q)}$ independent samples according to $\R'\sim D^{\set{q-2}}$. So, $\mathbf{D}_1=D^{\set{q-2}}$. 
For a fixed column $h$ we have the following: 
with probability $1-12\exp{(-\delta'q)}$ over the random choices, the composed ensemble 
\[ \widetilde{\mathcal{R}} = \bigl(R'_1, R'_2,\ldots, R'_{s})\]
is good wrt $h$, i.e., $\widetilde{\mathcal{R}}$ is a $(2\gamma, \gamma^2 - \delta')$-sampler for $\R'$ wrt the column $h$, and furthermore $\forall i \in [s]$,
\begin{equation*}  
\Bigl|\bigl|h(R_i')\bigr|-\E{\Bigl[\bigl|h(\R')\bigr|\Bigr]}\Bigr| \leq 2\gamma q.
\end{equation*}  

Hence, on expectation only $12\exp{(-\delta'q)}n$ columns are bad. Therefore, with probability at least $1/2$ at most $24\exp{(-\delta' q)}n$ columns are bad. Also, the probability of an $R_i'\in \mathcal{R'}$ having size less than $q-2$ (because some two of our $q-2$ choices of rows picked the same row)  is at most $q^2/m$. Thus, by the union bound the probability of (b) not holding is at most $s\cdot q^2/m$ which is less than $1/2$. Therefore, there is choice of $\widetilde{\mathcal{R}}$, which we call $\mathcal{R}^1$, such that at most $24\exp{(-\delta' q)}n$ columns are bad and (b) holds. The set of bad columns is $\mathcal{H}_2(1)$ and the set of good columns is $\mathcal{H}_1(1)$. Then, clearly (d) and (e) are true. 

Let $\mathcal{H}_2(1)=\set{h_1,\ldots,h_b}$ where $b\leq 24\exp{(-\delta' q)}n$ and WLOG assume that $1$ is the most frequent symbol in $h_1$. Retain only those rows in $U$ that correspond to the symbol $1$ in $h_1$. Call this pruned universe $U'$: we have ensured that so long as we add rows to $R_i\in \mathcal{R}^1$ only from $U'$, the image size in $h_1$ is at most $h_1(R)+1\leq q-1$. Thus, by taking a multiplicative hit of at most $1/q$ we have rendered $h_1$ ineffective. Iterating this over $\mathcal{H}_2(1)$ we take a multiplicative hit of $\left(\frac{1}{q}\right)^{b}$. Hence, we obtain a universe $U'$, which will be $U_1$, such that $\log_2{|U'|}=\log_2{|U_1|}\geq \log_2{m}-24\exp{(-\delta' q)}n\log_2{q}$. This establishes (c) and (f). 
This establishes the claims for $k=1$; the induction step in general is similar.

Now, as our IH let us assume that for $(k-1)$ we have the partition $\mathcal{H}_1(k-1)\sqcup\mathcal{H}_2(k-1)$, $U_{k-1}\subseteq[m]$, $\mathcal{R}^{k-1}$, integers $(g_1,\ldots,g_{k-1})$ and $\mathbf{D}_{k-1}$ such that (a) through (f) are satisfied. 
Then, we repeat the above argument. We have $g_{k}=q-\cc(g_{k-1})-k\gamma q -2$.
We use Lemma~\ref{Composition_Lemma} for $h\in \mathcal{H}_1(k-1)$ with $\R$ being $\mathbf{D}_{k-1}$, and $\mathcal{R}=\mathcal{R}^{k-1}$ which is a $(k\gamma,\gamma^2-(k-1)\delta')$-sampler for $\mathbf{D}_{k-1}$ wrt $h$. Let $(R_1,\ldots,R_s)$ be obtained by $s=\exp{(\gamma^2-k\delta')}$ independent samples from $\mathcal{R}^{k-1}$.
Let $D$ be the uniform distribution over $U_{k-1}$ and let $\mathcal{R}'=(R'_1, R'_2, \ldots, R'_{s})$ be obtained by taking $s$ independent samples according to $\R'\sim D^{\set{g_k}}$.
We let $\mathbf{D}_k =\mathbf{D}_{k-1} \vee D^{\set{g_k}}$.
For a fixed column $h$ we have the following:
wp $1-12\exp{(-\delta'q)}$ over the random choices, the composed ensemble 
\[ \widetilde{\mathcal{R}} = \bigl(R_1\cup R'_1, R_2\cup R'_2,\ldots, R_s\cup R'_{s})\]
is good wrt $h$, i.e., $\widetilde{\mathcal{R}}$ is a $((k+1)\gamma,\gamma^2-k\delta')$-sampler for $\mathbf{D}_k$ wrt $h$, and furthermore $\forall i \in [s]$,
\begin{equation*}  
\Bigl|\bigl|h(R_i\cup R_i')\bigr|-\E{\Bigl[\bigl|h(\mathbf{D}_{k})\bigr|\Bigr]}\Bigr| \leq (k+1)\gamma q.
\end{equation*}  

Hence, on expectation only $12\exp{(-\delta'q)}n$ columns of $\mathcal{H}_1(k-1)$ are bad. Therefore, with probability at least $1/2$ at most $24\exp{(-\delta' q)}n$ columns of $\mathcal{H}_1(k-1)$ are bad. Also, the probability of an $R_i\cup R_i'\in \widetilde{\mathcal{R}} $ having size less than $g_{\leq k}$ (because some two of our $q-\cc(g_{k-1})-k\gamma q -2$ choices of rows for $R_i'$ picked the same row of collided with some row in $R_i$)  is at most $(q\ln q)^2/|U_{k-1}|$. Thus, by the union bound the probability of (b) not holding is at most $s\cdot (q\ln{q})^2/|U_{k-1}|$ which is less than $1/2$. Therefore, there is choice of $\widetilde{\mathcal{R}}$, which we call $\mathcal{R}^k$, such that at most $24\exp{(-\delta' q)}n$ columns of $\mathcal{H}_1(k-1)$ are bad and (b) holds. Combining these bad columns with $\mathcal{H}_2(k-1)$  we obtain $\mathcal{H}_2(k)$ and the columns not in $\mathcal{H}_2(k)$ form the set $\mathcal{H}_1(k) = \mathcal{H}_2(k)$. Then, clearly (d) and (e) are true. 

Let $\mathcal{H}_2(k)\setminus \mathcal{H}_2(k-1)=\set{h_1,\ldots,h_b}$ where $b\leq 24\exp{(-\delta' q)}n$ and WLOG assume that $1$ is the most frequent symbol in $h_1$. Retain only those rows in $U_{k-1}$ that correspond to the symbol $1$ in $h_1$. Call this pruned universe $U'$: this pruning ensures that so long as we add rows to $R_i\in R^k$ only from $U'$, the image size in $h_1$ is at most $q-1$. Thus, by taking a multiplicative hit of at most $1/q$ we have rendered $h_1$ ineffective. Iterating this over $\mathcal{H}_2(k)$ we take a multiplicative hit of $\left(\frac{1}{q}\right)^{b}$. Hence, we obtain a universe $U'$, which will be $U_k$, such that $\log_2{|U'|}=\log_2{|U_k|}\geq \log_2{|U_{k-1}|}-24\exp{(-\delta' q)}n\log_2{q}\geq \log_2{m} - k\log_2{q}\cdot 24\exp{(-\delta'q)}n $. Together with property (c) of $U_{k-1}$ this establishes (c) and (f). 
This completes the induction step.
\end{proof}

\begin{proof}[Proof of Theorem~\ref{Result1} (main result of the paper)]
Fix an $\epsilon' < 1/6$ and let $C$ be an $\epsilon'q\ln{q}$-list-decoding-code for the $q/(q-1)$ channel. Choose $\lambda \ll \epsilon'$ and let $\epsilon =\epsilon' + \lambda$. Let $q$ be sufficiently large so that $k= q^{\epsilon}\geq e q^{\epsilon'+\lambda/2}$. 
We will appeal to Lemma~\ref{Iteration_Lemma} (with $k$ and $\epsilon$) and assume that $n\leq \exp{(\delta'q)}\log_2{m}/(48\cdot q^\epsilon \log_2{q})$. Then, by choosing a set of rows $R$ in the ensemble $\mathcal{R}^k$ and using (b) and Proposition~\ref{Target_Proposition} we obtain that $|R|\geq \epsilon'q\ln{q} $. However, using (c) we have that for all columns $h\in \mathcal{H}_2(k)$, $|h(R)|\leq q-1$. Also, using (e) and Lemma~\ref{PCC_Lemma} we obtain that for all $h\in \mathcal{H}_1(k)$, $|h(R)|<q$. This is a contradiction and hence $n>\exp{(\delta'q)}\log_2{m}/(48\cdot q^\epsilon \log_2{q})$ or for sufficiently large $q$ we have $n>\Omega(\exp{(q^{1-6\epsilon'}/8)}\log_2{m})$.

We note that it is possible by a more careful analysis to improve the bound of $\Omega(\exp{(q^{1-6\epsilon'}/8)}\log_2{m})$ to $\Omega(\exp{(q^{1-4\epsilon'}/8)}\log_2{m})$ in which case we may apply the bound till a list size of $q\ln{q}/4$. This bound is obtained by modifying Lemma~\ref{Iteration_Lemma} to accommodate $\gamma'\tand \delta'$ which vary across the induction steps and being more scrupulous about the argument in the preceding paragraph.

\end{proof}

\section*{Acknowledgements}

We are grateful to Prahladh Harsha for the numerous detailed discussions that led to the result reported in this paper, and also for proof-reading it. We also thank Ramprasad Saptharishi for his help with Lemma \ref{PCC_Lemma}. 

\bibliographystyle{IEEEtran}
\bibliography{bibi}

\begin{thebibliography}{1}
\providecommand{\url}[1]{#1}
\csname url@samestyle\endcsname
\providecommand{\newblock}{\relax}
\providecommand{\bibinfo}[2]{#2}
\providecommand{\BIBentrySTDinterwordspacing}{\spaceskip=0pt\relax}
\providecommand{\BIBentryALTinterwordstretchfactor}{4}
\providecommand{\BIBentryALTinterwordspacing}{\spaceskip=\fontdimen2\font plus
\BIBentryALTinterwordstretchfactor\fontdimen3\font minus
  \fontdimen4\font\relax}
\providecommand{\BIBforeignlanguage}[2]{{%
\expandafter\ifx\csname l@#1\endcsname\relax
\typeout{** WARNING: IEEEtran.bst: No hyphenation pattern has been}%
\typeout{** loaded for the language `#1'. Using the pattern for}%
\typeout{** the default language instead.}%
\else
\language=\csname l@#1\endcsname
\fi
#2}}
\providecommand{\BIBdecl}{\relax}
\BIBdecl

\bibitem{E88}
P.~Elias, ``Zero error capacity under list decoding,'' \emph{IEEE Transactions
  on Information Theory}, vol.~34, no.~5, pp. 1070--1074, Sep 1988.

\bibitem{DGR2016}
\BIBentryALTinterwordspacing
M.~Dalai, V.~Guruswami, and J.~Radhakrishnan, ``An improved bound on the
  zero-error list-decoding capacity of the 4/3 channel,'' in \emph{2017 {IEEE}
  International Symposium on Information Theory, {ISIT} 2017, Aachen, Germany,
  June 25-30, 2017}, 2017, pp. 1658--1662. [Online]. Available:
  \url{https://doi.org/10.1109/ISIT.2017.8006811}
\BIBentrySTDinterwordspacing

\bibitem{Arikan94}
E.~Arikan, ``An upper bound on the zero-error list-coding capacity,''
  \emph{IEEE Transactions on Information Theory}, vol.~40, no.~4, pp.
  1237--1240, Jul 1994.

\bibitem{KM88}
\BIBentryALTinterwordspacing
J.~Korner and K.~Marton, ``New bounds for perfect hashing via information
  theory,'' \emph{European Journal of Combinatorics}, vol.~9, no.~6, pp. 523 --
  530, 1988. [Online]. Available:
  \url{http://www.sciencedirect.com/science/article/pii/S0195669888800489}
\BIBentrySTDinterwordspacing

\bibitem{FK84}
\BIBentryALTinterwordspacing
M.~L. Fredman and J.~Komlós, ``On the size of separating systems and families
  of perfect hash functions,'' \emph{SIAM Journal on Algebraic Discrete
  Methods}, vol.~5, no.~1, pp. 61--68, 1984. [Online]. Available:
  \url{https://doi.org/10.1137/0605009}
\BIBentrySTDinterwordspacing

\bibitem{CRRS}
\BIBentryALTinterwordspacing
S.~Chakraborty, J.~Radhakrishnan, N.~Raghunathan, and P.~Sasatte, ``Zero error
  list-decoding capacity of the \emph{q}/(\emph{q}-1) channel,'' in
  \emph{{FSTTCS} 2006: Foundations of Software Technology and Theoretical
  Computer Science, 26th International Conference, Kolkata, India, December
  13-15, 2006, Proceedings}, 2006, pp. 129--138. [Online]. Available:
  \url{https://doi.org/10.1007/11944836_14}
\BIBentrySTDinterwordspacing

\end{thebibliography}

\appendix

\proofofphasedcc

\proofofTargetProposition

\end{document}